\pdfoutput=1
\documentclass[10pt,a4paper]{article}
\usepackage[utf8]{inputenc}
\usepackage[T1]{fontenc}
\usepackage{authblk}
\usepackage{setspace}
\usepackage[margin=1.25in]{geometry}
\usepackage{graphicx}
\usepackage{amsmath}
\usepackage{amsthm}

\usepackage{pxfonts}           % preserves KL1 special arrow glyphs
\usepackage[numbers,sort&compress]{natbib}
\usepackage{booktabs}
\usepackage{array}
\usepackage{enumitem}
\usepackage[labelsep=colon,font=normalsize,figurename=Fig.,tablename=Table,skip=8pt]{caption}
\usepackage{microtype}
\usepackage[hidelinks]{hyperref}

\hypersetup{
  pdftitle={A Counterexample to KL1 Completeness and a Complete Disjunctive Calculus},
  pdfauthor={Damiano Presa, Fabio Aurelio D'Asaro, Massimiliano Badino},
  pdfkeywords={Kant, conditional imperatives, permissives, input/output logic,
               completeness, Agda}}
\title{A Counterexample to KL$_1$ Completeness\\
       and a Complete Disjunctive Calculus}
\author{Damiano Presa$^{1}$, Fabio Aurelio D'Asaro$^{1}$, and Massimiliano Badino$^{1}$\\\small $^{1}$Department of Human Sciences, University of Verona, Verona, Italy}
\date{}

\setlist{itemsep=0pt,parsep=0pt,topsep=4pt plus 2pt,leftmargin=*}
\setlist[itemize]{label=--}
\newtheoremstyle{spplain}{8pt plus 2pt minus 2pt}{8pt plus 2pt minus 2pt}%
  {\itshape}{0pt}{\bfseries}{}{.5em}{\thmname{#1}\thmnumber{ #2}\thmnote{ (#3)}}
\newtheoremstyle{spdefinition}{8pt plus 2pt minus 2pt}{8pt plus 2pt minus 2pt}%
  {\normalfont}{0pt}{\bfseries}{}{.5em}{\thmname{#1}\thmnumber{ #2}\thmnote{ (#3)}}
\newtheoremstyle{spremark}{8pt plus 2pt minus 2pt}{8pt plus 2pt minus 2pt}%
  {\normalfont}{0pt}{\itshape}{}{.5em}{\thmname{#1}\thmnumber{ #2}\thmnote{ (#3)}}
\theoremstyle{spplain}
\newtheorem{theorem}{Theorem}
\newtheorem{proposition}[theorem]{Proposition}
\newtheorem{lemma}[theorem]{Lemma}
\newtheorem{corollary}[theorem]{Corollary}
\newtheorem*{conjecture}{Conjecture}
\theoremstyle{spdefinition}
\newtheorem{definition}[theorem]{Definition}
\newtheorem{example}[theorem]{Example}
\theoremstyle{spremark}
\newtheorem{remark}[theorem]{Remark}

\newcommand{\must}{\boxright}
\newcommand{\may}{\Diamondright}
\newcommand{\anyto}{\mathrel{\circ\mkern-7mu\rightarrow}}
\newcommand{\KL}[1]{\mathrm{KL}_{#1}}
\newcommand{\klone}{\texorpdfstring{KL$_1$}{KL1}}
\newcommand{\kltwo}{\texorpdfstring{KL$_2$}{KL2}}
\newcommand{\klthree}{\texorpdfstring{KL$_3$}{KL3}}
\newcommand{\kld}{KL$_1^{\vee}$}
\newcommand{\kltwod}{KL$_2^{\vee}$}
\newcommand{\entails}{\models_{\KL{1}}}
\newcommand{\nentails}{\not\models_{\KL{1}}}
\newcommand{\entailstwo}{\models_{\KL{2}}}
\newcommand{\derives}{\vdash_{\KL{1}}}
\newcommand{\nderives}{\nvdash_{\KL{1}}}
\newcommand{\derivesd}{\vdash^{\vee}_{\KL{1}}}
\newcommand{\nderivesd}{\nvdash^{\vee}_{\KL{1}}}

\newcommand{\out}{\mathit{out}_1}
\newcommand{\outtwo}{\mathit{out}_2}
\newcommand{\cns}{\mathit{cns}}
\newcommand{\stp}{\mathit{step}}
\newcommand{\kl}{\mathit{kl}_1}
\newcommand{\kltwofn}{\mathit{kl}_2}
\newcommand{\deriv}{\mathit{deriv}_1}
\newcommand{\derivd}{\mathit{deriv}^{\vee}_1}
\newcommand{\derivtwo}{\mathit{deriv}_2}
\newcommand{\derivtwod}{\mathit{deriv}^{\vee}_2}
\newcommand{\aux}{\mathit{aux}}
\newcommand{\Valid}{\mathit{Valid}}
\newcommand{\Perm}{\mathit{Perm}}
\newcommand{\nx}{\mathit{succ}}
\newcommand{\Atoms}{\mathcal{A}}
\newcommand{\Lits}{\mathcal{L}}
\newcommand{\Cons}{\mathcal{C}}
\newcommand{\rn}[1]{\textup{\textsc{#1}}}
\newcommand{\agda}[1]{\texttt{#1}}
\newcommand{\Rz}{R_0}
\newcommand{\rz}{r_0}
\newcommand{\cmp}[1]{\overline{#1}}
\newcommand{\nt}{\mathord{\sim}}
\newcommand{\infr}[4]{\mbox{\rn{#1}}\;\dfrac{\;#2\;}{\;#3\;}\;#4}

\begin{document}
\maketitle
\begin{abstract}
The logic \klone{} of Evans, Sergot and Stephenson formalizes Kant's rules as
conditional imperatives $B\must C$ and permissives $B\may C$ over acts.
Its consequence relation is operational: a rule follows from a set of rules
when adjoining it leaves the set of outcomes unchanged for every input. The
authors proved their calculus sound and conjectured that it is complete. We
show that the conjecture fails: for distinct atoms $b,c,d$ the rules
$\top\may b\vee d$, $b\may c$, $d\may c$ and $c\must b\vee d$ entail
$\top\may c$, which is not derivable. Non-derivability is proved by an
invariant that holds at every derivation depth and for arbitrary
intermediate atoms. The obstruction is the side condition of
\rn{may-trans}, which requires each new act to bring with it a single
intermediary act. Replacing it by a disjunctive side condition yields
\rn{may-trans-disj}, a strict generalization from which the original rule
is derivable. Using a local characterization of entailment, we prove the
repaired calculus sound for arbitrary sets of rules and complete for finite
ones, and we show that, with finite inputs, completeness cannot be extended
to infinite sets. The main results on \klone{} are formalized in Agda.
For \kltwo{} we observe
that the published reduction to \klone{} passes through infinite rule sets;
re-establishing the decomposition for finite rule sets, we prove that \kltwo{}
with the disjunctive rule is sound and complete for finite sets of rules. We close with remarks on
\klthree{} and on the significance of completeness for the relation between
the operational and the inferential questions of Kant's logic.
\end{abstract}
\medskip
\noindent\textbf{Keywords:}\enspace Kant $\cdot$ Conditional
imperatives $\cdot$ Permissives $\cdot$ Input/output logic $\cdot$
Completeness $\cdot$ Formal verification
\medskip

% ---------------------------------------------------------------------------
\section{Introduction}
% ---------------------------------------------------------------------------

Evans, Sergot and Stephenson~\cite{ESS20} introduced a family of logics,
\klone, \kltwo{} and \klthree, designed to formalize Kant's conception of a
rule as a general procedure relating \emph{acts} rather than propositions.
In \klone, the base system, a rule is either a conditional imperative
$B\must C$ (``if you are doing all of $B$, do one of $C$!'') or a
conditional permissive $B\may C$ (``if you are doing all of $B$, feel free
to do any of $C$!''). Given a set $R$ of rules and a set $A$ of acts already
being performed, the function $\out(R,A)$ collects the alternative packages
of acts that can be generated from $A$ by applying the rules and that
satisfy every imperative. The logic belongs, in a broad sense, to the family
of input/output logics~\cite{MvdT00}.

On this basis Evans, Sergot and Stephenson (henceforth ESS) distinguish two
questions~\cite[\S\S2.4, 3.6]{ESS20}. The \emph{operational} question asks
which further acts may or must be performed, given the rules adopted and the
acts already performed; it is answered by $\out$. The \emph{inferential}
question asks which further rules may or must be adopted, given the rules
already adopted. It is answered semantically by an entailment relation
$\entails$, defined by strong equivalence of outcome sets, and syntactically
by the inference rules of their Figure~2. ESS prove that the calculus is
sound (their Proposition~5) and state its completeness as
Conjecture~1, supported by extensive randomized testing.

This paper settles the conjecture negatively and repairs the calculus. Our
contributions are as follows.
\begin{enumerate}
\item A counterexample with four rules over three atoms (Section~\ref{sec:counterexample}):
  $\{\top\may b\vee d,\ b\may c,\ d\may c,\ c\must b\vee d\}$ entails
  $\top\may c$ in the sense of ESS's Definition~5, but $\top\may c$ has no
  derivation. Non-derivability is proved by an invariant over arbitrary
  derivations, not by search.
\item A repair (Section~\ref{sec:rule}). The rule \rn{may-trans} requires,
  for every new act $c$, an imperative $A\wedge c\must b_i$ with a
  \emph{single} intermediary $b_i$. We replace it by
  \rn{may-trans-disj}, which only requires the disjunctive imperative
  $A\wedge c\must b_1\vee\dots\vee b_n$. The original rule is derivable in
  the repaired calculus, which derives strictly more.
\item Soundness for arbitrary sets of rules (Section~\ref{sec:soundness}),
  via a local characterization of entailment that also shows that
  $\entails$ depends only on outcome sets and that the set of consequences
  of $R$ is closed under the inference rules.
\item Completeness for finite sets of rules (Section~\ref{sec:completeness}),
  together with an example showing that, for finite inputs, no sound
  calculus whose derivations are finite can be complete for infinite sets
  of rules.
\item A mechanization in the proof assistant Agda of the main results on
  \klone{} just listed (Section~\ref{sec:mechanization}).
\item An analysis of the consequences for \kltwo{} and \klthree{}
  (Section~\ref{sec:kl2} and Appendix~\ref{app:kl2}). The reduction of
  \kltwo{} to \klone{} used by ESS passes through infinite rule sets. For
  finite rule sets we re-establish it and prove that \kltwo{} with the
  disjunctive rule is sound and complete; these proofs are not mechanized
  and still require independent review.
\end{enumerate}
Section~\ref{sec:philosophy} discusses the philosophical significance of
these results, and Section~\ref{sec:conclusion} concludes.

% ---------------------------------------------------------------------------
\section{Preliminaries}
\label{sec:prelim}
% ---------------------------------------------------------------------------

We recall the definitions of~\cite[\S3]{ESS20}, keeping their notation.

\subsection{Syntax}

Let $\Atoms$ be the set of all atoms; $a,b,c,t$ range over atoms. Bodies and
heads of rules, inputs and outcomes are \emph{finite} sets of atoms, ranged
over by $A,B,C,S,T,X,Y$.\footnote{ESS speak of sets of atoms; finiteness of
bodies and heads is implicit in their reading of bodies as conjunctions and
heads as disjunctions, and inputs are explicitly finite
in~\cite[\S3.2]{ESS20}. The choice of $\Atoms$ is immaterial for our
results; in the Agda formalization $\Atoms$ is the set of natural numbers.}
A \emph{rule} is an imperative $B\must C$ or a permissive $B\may C$, with
\emph{body} $B$ (read as a conjunction) and \emph{head} $C$ (read as a
disjunction). We write $B\anyto C$ for a rule of either kind. Following ESS
we write bodies as conjunctions and heads as disjunctions, $\top$ for the
empty body and $\bot$ for the empty head, and $A\wedge b$ for
$A\cup\{b\}$; thus $p\wedge q\must r\vee s$ is the rule
$\{p,q\}\must\{r,s\}$. Sets $R$ of rules may be infinite unless stated
otherwise.

\subsection{Operational semantics}

\begin{definition}[Satisfaction {\cite[Def.~1]{ESS20}}]
A set $X$ of atoms \emph{satisfies} a rule, written $X\models r$, as follows:
$X\models B\must C$ if $B\not\subseteq X$ or $C\cap X\neq\emptyset$, and
$X\models B\may C$ always. $X\models R$ if $X\models r$ for every $r\in R$;
we then call $X$ a \emph{model} of $R$.
\end{definition}

\begin{definition}[Outcomes {\cite[Def.~2]{ESS20}}]
For a set $R$ of rules and a finite set $A$ of atoms,
\begin{align*}
\stp(R,X) &= \{\,c \mid B\anyto C\in R,\ B\subseteq X,\ c\in C\,\},\\
\cns_0(R,A) &= \{A\},\qquad
\cns_{n+1}(R,A) = \{\,X\cup\{t\}\mid X\in\cns_n(R,A),\ t\in\stp(R,X)\,\},\\
\cns(R,A) &= \textstyle\bigcup_{n\ge 0}\cns_n(R,A),\qquad
\out(R,A) = \{\,X\in\cns(R,A)\mid X\models R\,\}.
\end{align*}
\end{definition}

Thus $X\in\cns(R,A)$ if and only if there is a \emph{trace}
$A=X_0,X_1,\dots,X_n=X$ in which each $X_{k+1}=X_k\cup\{t_k\}$ for some
$t_k\in\stp(R,X_k)$. Both kinds of rules generate acts; only imperatives
constrain outcomes. A step may add an atom that is already present.

\begin{lemma}[Traces]\label{lem:traces}
Let $R\subseteq R'$ be sets of rules and $A,X,Y,Y'$ finite sets of atoms.
\begin{enumerate}[label=(\alph*)]
\item If $X\in\cns(R,A)$ then $A\subseteq X$ and $X\in\cns(R',A)$.
\item If $Y\in\cns(R,A)$ and $X\in\cns(R,Y)$ then $X\in\cns(R,A)$.
\item If $X\in\cns(R,Y)$ and $Y\subseteq Y'\subseteq X$ then $X\in\cns(R,Y')$.
\item $X\models R$ if and only if $X\in\out(R,X)$.
\end{enumerate}
\end{lemma}

\begin{proof}
(a) and (b) are immediate. For (c), if $Y=X_0,\dots,X_n=X$ is a trace, so is
$X_0\cup Y',\dots,X_n\cup Y'$: every body included in $X_k$ is included in
$X_k\cup Y'$, because bodies are positive, and
$(X_k\cup Y')\cup\{t_k\}=X_{k+1}\cup Y'$. It starts at $Y'$ and ends at
$X\cup Y'=X$. For (d), $X\in\cns_0(R,X)$; this is~\cite[Prop.~2]{ESS20}.
\end{proof}

\subsection{Entailment}

\begin{definition}[Entailment {\cite[Def.~5]{ESS20}}]\label{def:entailment}
A set $R$ of rules \emph{entails} a rule $r$, written $R\entails r$, if
$\out(R,A)=\out(R\cup\{r\},A)$ for every finite set $A$ of atoms. We write
$\kl(R)=\{r\mid R\entails r\}$.
\end{definition}

This is not model-theoretic consequence. For instance, $\emptyset\nentails
\top\may p$, although every set satisfies $\top\may p$: adjoining the
permissive creates the new outcome $\{p\}$ from the input $\emptyset$. ESS
take inputs to be finite~\cite[\S3.2]{ESS20}, while the wording of their
Definition~5 quantifies over all sets of atoms; for finite $R$ the two
readings coincide (Remark~\ref{rem:inputs}), and our counterexample holds
under both.

\subsection{The inference rules}

The inference rules of \klone{} are displayed in Fig.~\ref{fig:rules}. The
set $\deriv(R)$ of rules \emph{derivable} from $R$ is the least set of rules
that contains $R$ and is closed under these rules, where the side condition
of \rn{may-trans} is read as the requirement that the displayed imperatives
$A\wedge c\must b_i$ belong to the set as well; we write $R\derives r$ for
$r\in\deriv(R)$.\footnote{This is the most permissive reading of the side
condition; our non-derivability result holds a fortiori for stricter
readings, and also when the side condition is read semantically
(Remark~\ref{rem:robust}). We read $\subset$ in \rn{must-si}, \rn{may-si}
and \rn{may-so} as proper inclusion; the non-strict reading only adds
trivial instances. Instances of \rn{must-trans} and \rn{may-trans} with
$n=0$ are likewise immaterial: the former are instances of
\rn{quod-libet}, and the latter can only conclude $A\may\bot$, which is
derivable by \rn{may-id} and \rn{may-so}.} Since all heads are finite, every
derivable rule has a derivation, i.e.\ a finite tree of rule applications. In
derivations we tacitly omit applications of \rn{must-si}, \rn{may-si} and
\rn{may-so} whose inclusion is an equality.

\begin{figure}[t]
\hrule\vspace{10pt}
\centering
$\infr{must-id}{-}{A\must A}{A\neq\emptyset}\qquad\quad
 \infr{must-si}{A\must B}{A'\must B}{A\subset A'}$\\[10pt]
$\infr{must-union}{A\must B\quad A\may C}{A\must B\cup C}{}$\\[10pt]
$\infr{must-trans}{A\must b_1\vee\dots\vee b_n\quad A\wedge b_1\must C\quad
 \dots\quad A\wedge b_n\must C}{A\must C}{}$\\[10pt]
$\infr{quod-libet}{A\must\bot}{A\must B}{}$\\[12pt]
$\infr{may-id}{-}{A\may A}{}\qquad\quad
 \infr{may-si}{A\may B}{A'\may B}{A\subset A'}$\\[10pt]
$\infr{may-union}{A\may B\quad A\may C}{A\may B\cup C}{}$\\[10pt]
$\begin{array}{@{}r@{}}
 \infr{may-trans}{A\may b_1\vee\dots\vee b_n\quad A\wedge b_1\may C\quad
 \dots\quad A\wedge b_n\may C}{A\may C}{}\\[6pt]
 \text{if for every $c\in C$, $A\wedge c\must b_i$ for some
 $b_i\in\{b_1,\dots,b_n\}$}
\end{array}$\\[10pt]
$\infr{may-so}{A\may B}{A\may B'}{B'\subset B}$\\[10pt]
$\infr{may-must}{A\must B}{A\may B}{}\qquad\quad
 \infr{may-falsum}{A\wedge b\must\bot}{A\may b}{}$
\caption{Inference rules of \klone{} (reproduced from~\cite[Fig.~2]{ESS20})}
\label{fig:rules}
\vspace{4pt}\hrule
\end{figure}

\begin{conjecture}[{\cite[Conj.~1]{ESS20}}]
\klone{} is complete: $R\entails r$ implies $R\derives r$.
\end{conjecture}

% ---------------------------------------------------------------------------
\section{A Counterexample to the Completeness of \klone}
\label{sec:counterexample}
% ---------------------------------------------------------------------------

Fix three distinct atoms $b,c,d$ and let
\[
\Rz=\{\ \top\may b\vee d,\quad b\may c,\quad d\may c,\quad c\must b\vee d\ \},
\qquad \rz=\top\may c .
\]
In words: one may do $b$ or $d$; if one does $b$ one may do $c$, and
likewise if one does $d$; and whoever does $c$ must do $b$ or $d$. The
candidate consequence $\rz$ says that one may do $c$ outright. The models of
$\Rz$ are the sets $X$ such that $c\in X$ implies $b\in X$ or $d\in X$.

\begin{proposition}[Outcomes of $\Rz$]\label{prop:outcomes}
For every finite set $A$ of atoms,
$\out(\Rz,A)=\out(\Rz\cup\{\rz\},A)=\mathcal{O}(A)$, where
\[
\mathcal{O}(A)=\{\,X \mid A\subseteq X\subseteq A\cup\{b,c,d\},\
\text{and } c\in X \text{ implies } b\in X \text{ or } d\in X\,\}.
\]
In particular $\Rz\entails\rz$.
\end{proposition}

\begin{proof}
Let $R$ be $\Rz$ or $\Rz\cup\{\rz\}$. Every step under $R$ adds $b$, $c$ or $d$, so each
$X\in\cns(R,A)$ satisfies $A\subseteq X\subseteq A\cup\{b,c,d\}$; if
moreover $X\models R$, then $X\models c\must b\vee d$, which is the
displayed implication. Hence $\out(R,A)\subseteq\mathcal{O}(A)$.
Conversely, let $X\in\mathcal{O}(A)$. Starting from $A$, use
$\top\may b\vee d$, which is always enabled, to add the elements of
$X\cap\{b,d\}$ that are not in $A$; this reaches
$Y=A\cup(X\cap\{b,d\})$. If $c\in A$ or $c\notin X$, then $Y=X$. Otherwise
$c\in X$, so $b\in X$ or $d\in X$, hence $b\in Y$ or $d\in Y$, and one of
$b\may c$, $d\may c$ adds $c$ to reach $Y\cup\{c\}=X$. In either case
$X\in\cns(\Rz,A)\subseteq\cns(\Rz\cup\{\rz\},A)$, and $X$ satisfies the only
imperative of both rule sets. Hence $\mathcal{O}(A)\subseteq\out(R,A)$.
\end{proof}

The argument applies verbatim to infinite inputs $A$. Note also that $\Rz$
is not inconsistent: $A\cup\{b\}\in\out(\Rz,A)$ for every $A$.

For non-derivability we use one property of derivable imperatives, which is
part of the soundness of \klone{} but is proved directly here.

\begin{lemma}[Derived imperatives are valid]\label{lem:valid}
If $R\derives A\must B$, then $B\cap X\neq\emptyset$ for every model $X$ of
$R$ with $A\subseteq X$.
\end{lemma}

\begin{proof}
By induction on the derivation. Imperatives are concluded only by membership
in $R$ and by \rn{must-id}, \rn{must-si}, \rn{must-union},
\rn{must-trans} and \rn{quod-libet}. Membership is immediate. For
\rn{must-id}, $A\neq\emptyset$ and $A\subseteq X$. \rn{must-si} and
\rn{must-union} follow from $A\subseteq A'$ and $B\subseteq B\cup C$. For
\rn{must-trans}, the first premise yields some $b_i\in X$, and then the
premise $A\wedge b_i\must C$ applies to $X$. For \rn{quod-libet}, the
premise $A\must\bot$ shows that no model of $R$ includes $A$.
\end{proof}

\begin{lemma}\label{lem:consequences}
\begin{enumerate}[label=(\alph*)]
\item $\Rz\nderives A\must\bot$ for every finite set $A$.
\item If $\Rz\derives A\wedge c\must t$ for an atom $t$, then $t\in A$ or
  $t=c$.
\end{enumerate}
\end{lemma}

\begin{proof}
(a) $A\cup\{b\}$ is a model of $\Rz$ that includes $A$, and it does not meet
$\emptyset$; apply Lemma~\ref{lem:valid}. (b) $A\cup\{b,c\}$ and
$A\cup\{c,d\}$ are models of $\Rz$ that include $A\cup\{c\}$. By
Lemma~\ref{lem:valid}, $t$ belongs to both, hence to their intersection,
which is $A\cup\{c\}$ because $b$, $c$, $d$ are distinct.
\end{proof}

\begin{proposition}[Invariant]\label{prop:invariant}
Let $A\anyto H$ be a rule of either kind with $\Rz\derives A\anyto H$. If
$c\in H$, then $A\cap\{b,c,d\}\neq\emptyset$.
\end{proposition}

\begin{proof}
By induction on the derivation. No restriction is placed on the atoms that
occur in bodies and heads; derivations may involve atoms other than
$b,c,d$.
\begin{itemize}
\item \emph{Members of $\Rz$.} The heads of $\top\may b\vee d$ and
  $c\must b\vee d$ do not contain $c$; the bodies of $b\may c$ and $d\may c$
  meet $\{b,c,d\}$.
\item \rn{must-id}, \rn{may-id}. The head equals the body, so $c\in A$.
\item \rn{must-si}, \rn{may-si}. The body of the premise is included in $A$.
\item \rn{must-union}, \rn{may-union}. If $c\in B\cup C$, then $c$ is in the
  head of one of the two premises (for \rn{must-union} possibly the
  permissive premise), whose body is $A$.
\item \rn{may-so}, \rn{may-must}. The head of the premise includes $H$ and
  its body is $A$.
\item \rn{quod-libet}, \rn{may-falsum}. Their premises $A\must\bot$ and
  $A\wedge b\must\bot$ are not derivable, by
  Lemma~\ref{lem:consequences}(a).
\item \rn{must-trans}. Suppose $c\in C$ and $c\notin A$. Then $A$ is itself
  a model of $\Rz$, so by Lemma~\ref{lem:valid} the premise
  $A\must b_1\vee\dots\vee b_n$ yields some $b_i\in A$. The induction
  hypothesis for the premise $A\wedge b_i\must C$ gives
  $(A\cup\{b_i\})\cap\{b,c,d\}\neq\emptyset$, and $A\cup\{b_i\}=A$.
\item \rn{may-trans}. Suppose $c\in C$. The side condition provides an $i$
  with $\Rz\derives A\wedge c\must b_i$. By
  Lemma~\ref{lem:consequences}(b), $b_i\in A$ or $b_i=c$. If $b_i\in A$, the
  induction hypothesis for $A\wedge b_i\may C$ applies as in the previous
  case. If $b_i=c$, then $c$ is in the head of the first premise
  $A\may b_1\vee\dots\vee b_n$, and the induction hypothesis for that
  premise applies.\qedhere
\end{itemize}
\end{proof}

\begin{theorem}[Incompleteness of \klone]\label{thm:incompleteness}
We have $\Rz\entails\rz$ but $\Rz\nderives\rz$. Hence Conjecture~1
of~\cite{ESS20}
is false, already for a set of four rules over three atoms.
\end{theorem}

\begin{proof}
Proposition~\ref{prop:outcomes}, and Proposition~\ref{prop:invariant} applied
to $\rz=\top\may c$, whose body is empty.
\end{proof}

\begin{remark}[Robustness]\label{rem:robust}
(i)~Restricted to empty bodies, Proposition~\ref{prop:invariant} says that
no derivable permissive $\top\may H$ has $c\in H$; this special case, which
suffices for Theorem~\ref{thm:incompleteness}, can also be verified by an
induction restricted to such rules. (ii)~The proof uses the side condition of
\rn{may-trans} only through Lemma~\ref{lem:consequences}(b), that is, only
through the validity of $A\wedge c\must b_i$ in the models of $\Rz$. It
therefore applies equally if the side condition requires membership in the
rule set, or semantic entailment $\Rz\entails A\wedge c\must b_i$ (by
Proposition~\ref{prop:local} below, entailed imperatives are valid).
(iii)~By Proposition~\ref{prop:outcomes}, entailment holds whether inputs
range over finite or arbitrary sets of atoms.
\end{remark}

The intended derivation of $\rz$ is easy to see: $\top\may b\vee d$ offers
$b$ or $d$, and each of them licenses $c$. \rn{may-trans} would chain these
permissions if, in addition, $c\must b$ or $c\must d$ were available. But
neither is entailed: $\{c,d\}$ and $\{b,c\}$ are models of $\Rz$, indeed
outcomes of $\Rz$ from $\emptyset$. Only the disjunctive imperative
$c\must b\vee d$ holds, and the side condition of \rn{may-trans} cannot use
it.\footnote{ESS also remark that a finite set of rules has only finitely
many inferential consequences~\cite[\S3.5]{ESS20}. This presupposes a finite
stock of atoms: \rn{may-id} alone yields $a\may a$ for every atom $a$.}

% ---------------------------------------------------------------------------
\section{The Disjunctive Transitivity Rule}
\label{sec:rule}
% ---------------------------------------------------------------------------

\begin{definition}[\rn{may-trans-disj}]\label{def:disj}
The rule \rn{may-trans-disj} is
\[
\begin{array}{@{}r@{}}
 \infr{may-trans-disj}{A\may b_1\vee\dots\vee b_n\quad A\wedge b_1\may C\quad
 \dots\quad A\wedge b_n\may C}{A\may C}{}\\[6pt]
 \text{if $n\ge 1$ and for every $c\in C$, $A\wedge c\must b_1\vee\dots\vee b_n$}
\end{array}
\]
The calculus \kld{} is obtained from Fig.~\ref{fig:rules} by replacing
\rn{may-trans} with \rn{may-trans-disj}. We write $\derivd(R)$ and
$R\derivesd r$ for derivability in \kld, the side condition again being read
as derivability of the imperatives $A\wedge c\must b_1\vee\dots\vee b_n$.
\end{definition}

The premises are those of \rn{may-trans}; only the side condition changes.
For each $c\in C$ it no longer asks for an imperative leading from $c$ back
to one particular intermediary $b_i$, but for a single imperative whose head
is the disjunction of all intermediaries. As in~\cite{ESS20}, this
disjunctive imperative is one rule with a disjunctive head; it is not a
disjunction ``$A\wedge c\must b_1$ or $\dots$ or $A\wedge c\must b_n$'' of
rules, a form that \klone{} cannot express. Different outcomes containing
$c$ may contain different intermediaries.

\begin{example}\label{ex:repair}
With $A=\emptyset$, $n=2$, $b_1=b$, $b_2=d$ and $C=\{c\}$, the members of
$\Rz$ yield
\[
\dfrac{\;\top\may b\vee d\qquad b\may c\qquad d\may c\;}{\top\may c}
\;\rn{may-trans-disj},
\]
whose side condition is the member $c\must b\vee d$ of $\Rz$. Hence
$\Rz\derivesd\rz$.
\end{example}

\begin{proposition}[Relation to \rn{may-trans}]\label{prop:generalization}
\begin{enumerate}[label=(\alph*)]
\item \rn{may-trans} is a derived rule of \kld: whenever its premises and
  side-condition imperatives belong to $\derivd(R)$, so does its conclusion.
  Consequently $\deriv(R)\subseteq\derivd(R)$ for every $R$, and \kld{}
  derives the same rules as the calculus of Fig.~\ref{fig:rules} extended by
  \rn{may-trans-disj}.
\item The inclusion is strict: $\rz\in\derivd(\Rz)\setminus\deriv(\Rz)$. In
  particular, \rn{may-trans-disj} is not admissible in the original calculus.
\item The instances of \rn{may-trans-disj} with $n=0$, which conclude
  $A\may C$ from $A\may\bot$ and $A\wedge c\must\bot$ for all $c\in C$, are
  derivable in \kld.
\end{enumerate}
\end{proposition}

\begin{proof}
(a) Consider an instance of \rn{may-trans} with $B=\{b_1,\dots,b_n\}$. If
$C=\emptyset$, its conclusion $A\may\bot$ is derivable by \rn{may-id} and
\rn{may-so}. Otherwise the side condition, applied to some $c\in C$, shows
$n\ge1$. Let $c\in C$ and let $A\wedge c\must b_i$ be derivable. From
$A\may B$ we obtain $A\wedge c\may B$ by \rn{may-si} (if $c\in A$ there is
nothing to do), and \rn{must-union} yields
$A\wedge c\must\{b_i\}\cup B$, that is $A\wedge c\must B$. Thus the side
condition of \rn{may-trans-disj} holds, and its premises are those of the
given instance. The remaining claims follow by induction on derivations.
(b)~Example~\ref{ex:repair} and Theorem~\ref{thm:incompleteness}. If the
original calculus were closed under \rn{may-trans-disj}, the instance of
Example~\ref{ex:repair}, whose premises are members of $\Rz$, would put $\rz$
in $\deriv(\Rz)$. (c)~For each $c\in C$, \rn{may-falsum} gives $A\may c$, and
\rn{may-union} combines these, starting from $A\may\bot$.
\end{proof}

The disjunctive rule therefore resolves the counterexample by a single
inference that uses only members of $\Rz$, and it does so without changing
the semantics: $\out$ and $\entails$ are exactly as in~\cite{ESS20}, and only
the proof system is modified. The soundness theorem of the next section
confirms that \rn{may-trans-disj} is valid for the original operational
semantics.

Reasoning through alternative intermediaries is needed because, in $\Rz$,
the act $c$ is always accompanied by an act that licenses it, but not always
by the same one. \rn{may-trans} can route a permission for $c$ only through
an intermediary that is contained in every model of the rules that contains
$A$ and $c$; the models $\{b,c\}$ and $\{c,d\}$ show that no such
intermediary exists.
The proof of completeness in Section~\ref{sec:completeness} shows that this
is the only missing pattern: for every body $A$, the set of acts that can be
added to $A$ in one step serves as an adequate disjunction of
intermediaries.

The repair keeps the protective role of the side condition that ESS discuss
in their \S3.5.1. Without any side condition, the rule would infer $p\may r$
from $p\may q$ and $q\may r$ (``visiting the south of France'' would license
``disrobing entirely'' via ``entering a naturist resort''). With
\rn{may-trans-disj}, this inference still requires the imperative
$p\wedge r\must q$, which is not available. The side condition is weakened
only from ``one $b_i$, contained in every model that contains $A$ and $c$''
to ``in every such model, some $b_i$''. It
is also worth noting that in the original
calculus the fan-out of \rn{may-trans} over several intermediaries is
redundant: since its side condition singles out one $b_i$ for each $c$,
every instance can be replaced by instances with $n=1$ together with
\rn{may-so} and \rn{may-union}. With the disjunctive side condition, the
disjunctive first premise does genuine work.

% ---------------------------------------------------------------------------
\section{Soundness}
\label{sec:soundness}
% ---------------------------------------------------------------------------

\subsection{A local characterization of entailment}

Definition~\ref{def:entailment} compares the outcome sets of $R$ and
$R\cup\{r\}$. The following characterization refers to $R$ alone.

\begin{definition}
For a set $R$ of rules and finite sets $B,C$ of atoms:
\begin{itemize}
\item $\Valid_R(B,C)$ holds if every model $X$ of $R$ with $B\subseteq X$
  satisfies $C\cap X\neq\emptyset$;
\item $\Perm_R(B,C)$ holds if $\out(R,Y\cup\{c\})\subseteq\out(R,Y)$ for
  every finite $Y\supseteq B$ and every $c\in C$.
\end{itemize}
\end{definition}

$\Perm_R(B,C)$ says that at any stage that includes $B$, adding an atom of
$C$ outright never produces an outcome that could not have been produced
anyway.

\begin{proposition}[Local characterization]\label{prop:local}
For every set $R$ of rules, finite or infinite, and all finite sets $B,C$:
\begin{enumerate}[label=(\alph*)]
\item $R\entails B\may C$ if and only if $\Perm_R(B,C)$;
\item $R\entails B\must C$ if and only if $\Valid_R(B,C)$ and $\Perm_R(B,C)$.
\end{enumerate}
\end{proposition}

\begin{proof}
Let $r$ be $B\may C$ or $B\must C$ and $R^+=R\cup\{r\}$.

($\Rightarrow$) Let $r$ be an imperative and $X$ a model of $R$ with
$B\subseteq X$. By Lemma~\ref{lem:traces}(d), $X\in\out(R,X)=\out(R^+,X)$,
so $X\models r$ and $C\cap X\neq\emptyset$. This proves $\Valid_R(B,C)$. For
$\Perm_R(B,C)$, let $Y\supseteq B$, $c\in C$ and
$X\in\out(R,Y\cup\{c\})$. Since $c\in\stp(R^+,Y)$,
Lemma~\ref{lem:traces}(a,b) gives $X\in\cns(R^+,Y)$. Moreover $X\models r$:
trivially if $r$ is a permissive, and by $\Valid_R(B,C)$, just proved, if
$r$ is an imperative. Hence $X\in\out(R^+,Y)=\out(R,Y)$.

($\Leftarrow$) By Lemma~\ref{lem:traces}(a), $\out(R,A)\subseteq\out(R^+,A)$,
using $\Valid_R(B,C)$ when $r$ is an imperative. Conversely, let
$X\in\out(R^+,A)$ have the trace $A=X_0,\dots,X_n=X$ under $R^+$. Then
$X\models R$. We show by downward induction on $k$ that
$X\in\cns(R,X_k)$, the case $k=n$ being trivial. If the step from $X_k$ to
$X_{k+1}$ uses a rule of $R$, then $X\in\cns(R,X_k)$ by
Lemma~\ref{lem:traces}(b). If it uses $r$, then $B\subseteq X_k$ and
$X_{k+1}=X_k\cup\{t\}$ with $t\in C$, so
$X\in\out(R,X_k\cup\{t\})\subseteq\out(R,X_k)$ by $\Perm_R(B,C)$. For $k=0$
we obtain $X\in\out(R,A)$.
\end{proof}

\begin{corollary}\label{cor:invariance}
Whether $R\entails r$ holds depends only on the function
$A\mapsto\out(R,A)$: if $\out(R,A)=\out(R',A)$ for all $A$, then
$\kl(R)=\kl(R')$. In particular $\kl(R\cup\{r\})=\kl(R)$ whenever
$R\entails r$.
\end{corollary}

\begin{proof}
$\Perm_R$ is defined from $\out(R,\cdot)$, and by
Lemma~\ref{lem:traces}(d), $\Valid_R(B,C)$ says that $C$ meets every $X$
with $B\subseteq X$ and $X\in\out(R,X)$.
\end{proof}

\subsection{Soundness of the repaired calculus}

\begin{theorem}[Soundness]\label{thm:soundness}
Let $R$ and $Q$ be sets of rules, finite or infinite, such that
$Q\subseteq\kl(R)$. If $Q\derivesd r$, then $R\entails r$. In particular,
$R\derivesd r$ implies $R\entails r$, and $\kl(R)$ is closed under the
inference rules of \kld{} and of Fig.~\ref{fig:rules}.
\end{theorem}

\begin{proof}
Call $r$ \emph{local} if $\Perm_R(B,C)$ holds when $r$ is $B\may C$, and both
$\Valid_R(B,C)$ and $\Perm_R(B,C)$ hold when $r$ is $B\must C$. By
Proposition~\ref{prop:local} the members of $Q$ are local, and it suffices to
show that every inference rule of \kld{} and of Fig.~\ref{fig:rules}
preserves locality. Several cases are direct: $\Valid_R(B,C)$ and
$\Perm_R(B,C)$ are preserved by enlarging $B$ (\rn{must-si}, \rn{may-si});
$\Perm_R(B,C)$ is preserved by shrinking $C$ (\rn{may-so}) and is pointwise
in $c\in C$ (\rn{must-union}, \rn{may-union}); $\Valid_R(A,B)$ implies
$\Valid_R(A,B\cup C)$ (\rn{must-union}); \rn{may-must} forgets validity. For
\rn{must-id} and \rn{may-id}, $c\in A\subseteq Y$ gives $Y\cup\{c\}=Y$, and
$A\neq\emptyset$ yields validity. Every $X\in\out(R,Y\cup\{c\})$ with
$A\subseteq Y$ is a model of $R$ that includes $A\cup\{c\}$. If the premise
$A\must\bot$ of \rn{quod-libet} is local, no model of $R$ includes $A$, so
the conclusion $A\must B$ is vacuously local; likewise the premise
$A\wedge b\must\bot$ of \rn{may-falsum} excludes models that include
$A\cup\{b\}$, so its conclusion $A\may b$ is vacuously local.

For \rn{must-trans}, validity is proved as in Lemma~\ref{lem:valid}. For
permission, let $Y\supseteq A$, $c\in C$ and $X\in\out(R,Y\cup\{c\})$.
Since $A\subseteq X$ and $X\models R$, the validity of
$A\must b_1\vee\dots\vee b_n$ gives some $b_i\in X$. By
Lemma~\ref{lem:traces}(c), $X\in\out(R,(Y\cup\{b_i\})\cup\{c\})$; the
permission part of $A\wedge b_i\must C$ at the stage $Y\cup\{b_i\}$ gives
$X\in\out(R,Y\cup\{b_i\})$; and the permission part of the first premise at
$Y$ gives $X\in\out(R,Y)$.

For \rn{may-trans-disj}, let $Y\supseteq A$, $c\in C$ and
$X\in\out(R,Y\cup\{c\})$. Now $A\cup\{c\}\subseteq X$, and the validity of
the side-condition imperative $A\wedge c\must b_1\vee\dots\vee b_n$ gives
some $b_i\in X$. The argument is then the same as for \rn{must-trans}, using
the premises $A\wedge b_i\may C$ and $A\may b_1\vee\dots\vee b_n$. The case of
\rn{may-trans} is identical, with the singleton side condition
$A\wedge c\must b_i$.
\end{proof}

The proof uses neither finiteness of $R$ nor the definite-program semantics
of~\cite[\S3.3]{ESS20}. The only property of traces it needs is
Lemma~\ref{lem:traces}(c): an intermediary $b_i$ that belongs to the final
outcome may be assumed from the start, because bodies are positive.

\begin{remark}[Soundness via least models]\label{rem:least-models}
For finite $R$, the soundness of a single application of
\rn{may-trans-disj} whose premises belong to $R$ can also be verified with
the least-model characterization of $\out$~\cite[Prop.~1]{ESS20}. Passing
from single inferences to derivations along the lines of the proof
of~\cite[Prop.~5]{ESS20} then requires that the outcome sets of $R$ and
$R\cup\{r_1,r_2\}$ coincide when $r_1$ and $r_2$ are separately entailed,
which the inductive step of that proof takes for granted. This
compositionality is not proved in~\cite{ESS20}; it follows from
Corollary~\ref{cor:invariance}, and Theorem~\ref{thm:soundness} avoids it
altogether. Moreover, the least-model characterization fails for some
infinite rule sets (Remark~\ref{rem:infinite}), so this approach does not
extend to infinite $R$, whereas Theorem~\ref{thm:soundness} does.
\end{remark}

% ---------------------------------------------------------------------------
\section{Completeness}
\label{sec:completeness}
% ---------------------------------------------------------------------------

Throughout this section $R$ is a finite set of rules. Let $H$ be the union of
the heads of the rules in $R$, a finite set, and for a finite set $S$ let
$\mu(S)=|H\setminus S|$. The \emph{successors} of a finite set $A$ are the
atoms that can be added to $A$ in one step:
\[
\nx(A)=\stp(R,A)\setminus A=\{\,t\notin A\mid B\anyto C\in R,\ B\subseteq A,\
t\in C\,\}\subseteq H .
\]
Adding a successor, or any element of $H\setminus S$, to $S$ decreases
$\mu$; this is the measure for all inductions below.

\begin{lemma}[Reachable models]\label{lem:reach}
If $X$ is a model of $R$ and $S\subseteq X$, then there is a model $T$ of
$R$ with $S\subseteq T\subseteq X$ and $T\in\cns(R,S)$, so that
$T\in\out(R,S)$.
\end{lemma}

\begin{proof}
By induction on $\mu(S)$. If $S\models R$, take $T=S$. Otherwise some
imperative $B\must C\in R$ has $B\subseteq S$ and $C\cap S=\emptyset$. Since
$B\subseteq X$ and $X\models R$, there is $t\in C\cap X$. Then
$S\cup\{t\}\in\cns(R,S)$, $S\cup\{t\}\subseteq X$ and
$\mu(S\cup\{t\})<\mu(S)$; the induction hypothesis and
Lemma~\ref{lem:traces}(b) give $T$.
\end{proof}

\begin{lemma}[Successors]\label{lem:succ}
Let $A$ be a finite set of atoms.
\begin{enumerate}[label=(\alph*)]
\item $R\derivesd A\may\nx(A)$.
\item If $T\in\cns(R,A)$ and $T\not\subseteq A$, then $T\cap\nx(A)\neq\emptyset$.
\end{enumerate}
\end{lemma}

\begin{proof}
(a) For each $B\anyto C\in R$ with $B\subseteq A$ we have $B\may C$ (by
\rn{may-must} if the rule is an imperative), hence $A\may C$ by
\rn{may-si} and $A\may C\setminus A$ by \rn{may-so}. Starting from $A\may\bot$
(by \rn{may-id}, and \rn{may-so} if $A\neq\emptyset$), \rn{may-union}
combines these finitely many permissives into $A\may\nx(A)$. (b) In a trace
$A=X_0,\dots,X_n=T$, let $k$ be least with $X_{k+1}\not\subseteq A$. Then
$X_k=A$, and the atom added at that step lies in $\stp(R,A)\setminus A$ and
in $T$.
\end{proof}

\begin{lemma}[Imperatives from permissives]\label{lem:must}
If $R\derivesd S\may D$ and $\Valid_R(S,D)$, then $R\derivesd S\must D$.
\end{lemma}

\begin{proof}
By induction on $\mu(S)$. If $S\models R$, then $\Valid_R(S,D)$ applied to
$X=S$ gives some $t\in D\cap S$. By \rn{must-id} and \rn{must-si},
$S\must t$, and \rn{must-union} with $S\may D$ gives
$S\must\{t\}\cup D$, that is $S\must D$. Otherwise some $B\must C\in R$ has
$B\subseteq S$ and $C\cap S=\emptyset$, and \rn{must-si} gives $S\must C$.
If $C=\emptyset$, \rn{quod-libet} gives $S\must D$. Otherwise let $t\in C$.
Then $t\in H\setminus S$, so $\mu(S\cup\{t\})<\mu(S)$; moreover
$S\wedge t\may D$ by \rn{may-si}, and $\Valid_R(S\cup\{t\},D)$ holds. By
the induction hypothesis $S\wedge t\must D$ for every $t\in C$, and
\rn{must-trans} yields $S\must D$.
\end{proof}

\begin{lemma}[Singleton permissives]\label{lem:atom}
Let $A$ be finite, $c\notin A$, and suppose that $\Perm_R(A,\{c\})$. Then
$R\derivesd A\may c$.
\end{lemma}

\begin{proof}
By induction on $\mu(A)$. Write $S^*=\nx(A)$. If $c\in S^*$, some
$B\anyto C\in R$ has $B\subseteq A$ and $c\in C$, and $A\may c$ follows by
\rn{may-must} (if needed), \rn{may-si} and \rn{may-so}. Suppose $c\notin S^*$.
\begin{enumerate}[label=(\roman*)]
\item By Lemma~\ref{lem:succ}(a), $A\may S^*$.
\item Let $b\in S^*$. Then $b\notin A$ and $b\in H$, so
  $\mu(A\cup\{b\})<\mu(A)$; also $c\notin A\cup\{b\}$, as $c\neq b$; and
  $\Perm_R(A\cup\{b\},\{c\})$ holds because $\Perm_R$ is preserved by
  enlarging the body. By the induction hypothesis, $A\wedge b\may c$.
\item $\Valid_R(A\cup\{c\},S^*)$: let $X$ be a model of $R$ with
  $A\cup\{c\}\subseteq X$. Lemma~\ref{lem:reach} gives
  $T\in\out(R,A\cup\{c\})$ with $T\subseteq X$, and $\Perm_R(A,\{c\})$ gives
  $T\in\out(R,A)$. As $c\in T\setminus A$, Lemma~\ref{lem:succ}(b) gives
  $T\cap S^*\neq\emptyset$, hence $X\cap S^*\neq\emptyset$.
\item From (i), \rn{may-si} gives $A\wedge c\may S^*$, and by (iii) and
  Lemma~\ref{lem:must}, $A\wedge c\must S^*$.
\end{enumerate}
If $S^*=\emptyset$, \rn{may-falsum} applied to $A\wedge c\must\bot$ yields
$A\may c$. Otherwise \rn{may-trans-disj}, with the premises (i) and (ii) and
the side condition (iv), yields $A\may c$.
\end{proof}

\begin{theorem}[Completeness for finite rule sets]\label{thm:completeness}
Let $R$ be a finite set of rules and $r$ a rule. If $R\entails r$, then
$R\derivesd r$. Consequently $\kl(R)=\derivd(R)$ for every finite $R$.
\end{theorem}

\begin{proof}
Let $r$ be $A\may C$ or $A\must C$. By Proposition~\ref{prop:local},
$\Perm_R(A,C)$ holds, hence $\Perm_R(A,\{c\})$ for each $c\in C$. For
$c\in A$, \rn{may-id} and \rn{may-so} give $A\may c$; for $c\notin A$,
Lemma~\ref{lem:atom} does. Combining these by \rn{may-union}, starting from
$A\may\bot$, gives $R\derivesd A\may C$. If $r$ is an imperative,
Proposition~\ref{prop:local} also gives $\Valid_R(A,C)$, and
Lemma~\ref{lem:must} yields $R\derivesd A\must C$. The equality follows with
Theorem~\ref{thm:soundness}.
\end{proof}

The proof is constructive, and the derivations it produces only involve
atoms that occur in $R$ or $r$. A familiar consequence is a model existence
theorem.

\begin{corollary}[Consistency and outcomes]\label{cor:model}
For a finite set $R$ of rules the following are equivalent:
(i)~$R\nderivesd\top\must\bot$; (ii)~$R$ has a model;
(iii)~$\out(R,\emptyset)\neq\emptyset$.
\end{corollary}

\begin{proof}
By Proposition~\ref{prop:local}, $R\entails\top\must\bot$ if and only if
$\Valid_R(\emptyset,\emptyset)$, since $\Perm_R(\emptyset,\emptyset)$ holds
trivially; that is, if and only if $R$ has no model;
Theorems~\ref{thm:soundness} and~\ref{thm:completeness} give (i)$\Leftrightarrow$(ii).
Lemma~\ref{lem:reach} with $S=\emptyset$ gives (ii)$\Rightarrow$(iii), and
outcomes are models.
\end{proof}

\begin{remark}[Infinite rule sets]\label{rem:infinite}
The restriction to finite $R$ cannot be dropped. Let $\Atoms$ be infinite and
$R_\infty=\{\top\must a\mid a\in\Atoms\}$. Outcomes are finite, but a model of
$R_\infty$ must contain every atom; so $\out(R_\infty,A)=\emptyset$ for every
finite $A$, and $R_\infty\entails\top\must\bot$. A derivation is a finite tree
and uses only finitely many members of $R_\infty$, forming a finite
$R'\subseteq R_\infty$. But the set of heads of $R'$ is an outcome of $R'$ from
$\emptyset$, so $R'\nentails\top\must\bot$, and by
Theorem~\ref{thm:soundness}, $R'\nderivesd\top\must\bot$. Hence
$R_\infty\nderivesd\top\must\bot$, and likewise in the original calculus. The
same example shows that the least-model characterization
of~\cite[Prop.~1]{ESS20} fails for $R_\infty$: the unique definite program
encoding $R_\infty$ consists of the facts $a\leftarrow{}$ for $a\in\Atoms$,
whose least model $\Atoms$ satisfies $R_\infty$, whereas
$\out(R_\infty,\emptyset)=\emptyset$. Under the reading of Definition~5 with
arbitrary inputs, $R_\infty$ is not a counterexample, since
$\out(R_\infty,\Atoms)=\{\Atoms\}$; we leave completeness for infinite rule
sets under that reading open.
\end{remark}

\begin{remark}[Finite and arbitrary inputs]\label{rem:inputs}
Let $R$ be finite and $L$ a finite set containing every atom of $R$ and $r$.
Since bodies and heads lie within $L$, traces from an arbitrary set $A$ and
from $A\cap L$ correspond step by step, and satisfaction of $R$ depends only
on atoms in $L$; hence
$\out(R,A)=\{X\cup(A\setminus L)\mid X\in\out(R,A\cap L)\}$, and likewise
for $R\cup\{r\}$. Thus for finite $R$, Definition~5 gives the same relation
whether inputs range over finite or over arbitrary sets of atoms, and
Theorems~\ref{thm:soundness} and~\ref{thm:completeness} hold under both
readings.
\end{remark}

\begin{remark}[Completeness via least models]\label{rem:alt-completeness}
The same strategy can also be carried out with the definite-program
semantics of~\cite[\S3.3]{ESS20}: reduction to singleton permissives,
induction on bodies, the successor set as the disjunction of
intermediaries, and \rn{may-trans-disj}, or \rn{may-falsum} when the
disjunctive imperative has an empty head. A proof along these lines relies
on two facts that are not established in~\cite{ESS20}. The first is the
least-model characterization of $\out$~\cite[Prop.~1]{ESS20}, which is
easily verified for finite $R$ but is only sketched in~\cite{ESS20}, and
which fails for some infinite $R$ (Remark~\ref{rem:infinite}). The second
is the transfer of semantic entailment along inference rules, for instance
from $R\entails A\may C$ to $R\entails A\may c$ for $c\in C$ by the
soundness of \rn{may-so}; this is the closure of $\kl(R)$ under the rules,
which does not follow from~\cite[Prop.~5]{ESS20} as stated but is provided
by Theorem~\ref{thm:soundness}.
\end{remark}

% ---------------------------------------------------------------------------
\section{Mechanization}
\label{sec:mechanization}
% ---------------------------------------------------------------------------

The main results on \klone{} in
Sections~\ref{sec:prelim}--\ref{sec:completeness}, listed in
Table~\ref{tab:agda}, are formalized in the proof assistant
Agda~\cite{Norell07}, version 2.8.0,
in safe mode, without libraries beyond Agda's builtin modules, and without
postulates, holes or disabled termination and positivity checks. The main
file, \agda{KL1Completeness.agda}, type-checks in about five seconds; a
supplementary file covers arbitrary inputs and the rule set $R_\infty$ of
Remark~\ref{rem:infinite}.

Atoms are natural numbers, with no bound on the atoms that may occur in
inputs or in intermediate steps of derivations. Finite sets are represented
canonically, so that extensional equality coincides with equality. A rule set
is a membership predicate, possibly infinite, and finiteness is an explicit
enumeration by a list. Satisfaction, traces, outcomes and entailment are
transcriptions of Definitions~1, 2 and~5 of~\cite{ESS20} (entailment as the two
inclusions of outcome sets for all finite inputs), and the original calculus
has exactly the rules of Fig.~\ref{fig:rules}, with the side condition of
\rn{may-trans} as derivations of singleton imperatives. Table~\ref{tab:agda}
lists the correspondence. The formal statements of soundness and
completeness are restated in a final summary module directly in terms of the
primitive notions.

\begin{table}[t]
\caption{Results and their Agda counterparts in \agda{KL1Completeness.agda}
(module prefixes as in the file; $^{\dagger}$ in \agda{KL1Supplement.agda})}
\label{tab:agda}
\footnotesize
\begin{tabular}{@{}>{\raggedright}p{0.36\textwidth}@{\hspace{6pt}}>{\raggedright\arraybackslash}p{0.61\textwidth}@{}}
\toprule
Result & Agda identifier\tabularnewline
\midrule
Lemma~\ref{lem:traces} & \agda{Semantics.cns-initial}, \agda{cns-mono},
  \agda{Locality.cns-concat}, \agda{cns-grow}\tabularnewline
Proposition~\ref{prop:outcomes} ($\Rz\entails\rz$) &
  \agda{ExampleEntailment.entailed}; \agda{entailed-all-inputs}$^{\dagger}$\tabularnewline
Lemma~\ref{lem:valid} & \agda{Obligations.must-valid}\tabularnewline
Proposition~\ref{prop:invariant} & \agda{NonDerivability.head-invariant}\tabularnewline
Theorem~\ref{thm:incompleteness} & \agda{Counterexample.counterexample},
  \agda{Summary.incompleteness}\tabularnewline
Proposition~\ref{prop:generalization} & \agda{Disj.may-trans-derivable},
  \agda{plus-to-disj}, \agda{may-trans-disj-any},\tabularnewline
 & \agda{DisjExample.example-disj}, \agda{disj-not-admissible-in-old}\tabularnewline
Proposition~\ref{prop:local} & \agda{Locality.sem-to-local},
  \agda{local-to-sem}\tabularnewline
Corollary~\ref{cor:invariance} & \agda{Invariance.entailment-invariant}\tabularnewline
Theorem~\ref{thm:soundness} & \agda{DisjSoundness.soundness},
  \agda{closure}; \agda{Soundness.sound-local-rel}\tabularnewline
Lemmas~\ref{lem:reach}--\ref{lem:atom} & \agda{Completeness.Complete.reach},
  \agda{may-succ}, \agda{first-step},\tabularnewline
 & \agda{must-complete}, \agda{may-atom}\tabularnewline
Theorem~\ref{thm:completeness} & \agda{DisjCompleteness.completeness},
  \agda{Summary.completeness}\tabularnewline
Remark~\ref{rem:infinite} (semantic part) &
  \agda{CountableScope.facts-entail-falsum}$^{\dagger}$,
  \agda{least-model-satisfies}$^{\dagger}$\tabularnewline
\bottomrule
\end{tabular}
\end{table}

The mechanization covers the statements, not the informal arguments: the
proofs in this paper follow the formal ones in outline but differ from them
in inessential details (for example, the entailment in
Proposition~\ref{prop:outcomes} is checked in Agda by a trace-replay
argument rather than by the explicit description of outcomes). Corollary~\ref{cor:model} is
an immediate consequence of mechanized results, but is not itself stated in
Agda. Remarks~\ref{rem:robust} and~\ref{rem:inputs} and the non-derivability
half of Remark~\ref{rem:infinite} are not mechanized. As
always, the guarantee rests on the correctness of Agda and on the fidelity of
the formal definitions to those of~\cite{ESS20}. Nothing in Section~\ref{sec:kl2} or Appendix~\ref{app:kl2}
is mechanized.

% ---------------------------------------------------------------------------
\section{Consequences for \kltwo{} and \klthree}
\label{sec:kl2}
% ---------------------------------------------------------------------------

\subsection{\kltwo}

We recall~\cite[\S4]{ESS20} briefly. \kltwo{} adds negated atoms $\nt a$;
literals form a set $\Lits$, $\cmp{\ell}$ is the complement of a literal
$\ell$, and $\cmp{C}=\{\cmp{c}\mid c\in C\}$. A set of literals is
\emph{consistent} if it contains no pair $a,\nt a$, and \emph{maximal
consistent} if, in addition, it contains $a$ or $\nt a$ for every atom $a$.
Rules have the same form as in \klone, with finite sets of literals as
bodies and heads, and \klone{} notions are applied to literals as if they
were atoms. Let $\Cons=\{a\wedge\nt a\must\bot\mid a\in\Atoms\}$. A set of
literals is \emph{violating} for $R$ if no maximal consistent set of
literals that includes it satisfies $R$; $\aux(R)$ consists of the rules
$A\setminus\{\ell\}\must\cmp{\ell}$ for finite violating sets $A$ and
$\ell\in A$; and $\outtwo(R,A)=\out(R\cup\Cons\cup\aux(R),A)$, with
$R\entailstwo r$ and $\kltwofn(R)$ defined from $\outtwo$ as in
Definition~\ref{def:entailment}~\cite[Defs.~6--8]{ESS20}. Let $v(R)$ be the
set of consistent sets $B\cup\cmp{C}$ with $B\must C\in R$, let $v^*(R)$ be
its closure under resolution, which passes from $A\cup\{\ell\}$ and
$B\cup\{\cmp{\ell}\}$ to $A\cup B$, and let $\aux_e(R)$ consist of the rules
$A\setminus\{\ell\}\must\cmp{\ell}$ for $A\in v^*(R)$ and
$\ell\in A$~\cite[Defs.~10--12]{ESS20}. The calculus of \kltwo{} adds to
Fig.~\ref{fig:rules} the rules $\nt$\textsc{-left}, with conclusion
$c\wedge\nt c\must\bot$, and $\nt$\textsc{-right}, which infers
$A\must\cmp{b}$ from $A\wedge b\must\bot$. We write \kltwod{} for \kltwo{}
with \rn{may-trans} replaced by \rn{may-trans-disj}, and $\derivtwod$ for its
derivability operator.

ESS establish the decomposition $\kltwofn(R)=\kl(R\cup\Cons\cup\aux(R))$
(Prop.~6), show that $\aux(R)$ may be replaced by $\aux_e(R)$
(Prop.~13), and prove that if \klone{} is complete then so is \kltwo{}
(Prop.~16), via
\begin{align*}
r\in\kltwofn(R)\ &\Rightarrow\ r\in\kl(R\cup\Cons\cup\aux_e(R))\\
&\Rightarrow\ r\in\deriv(R\cup\Cons\cup\aux_e(R))\ \Rightarrow\ r\in\derivtwo(R).
\end{align*}
One might expect that Theorem~\ref{thm:completeness}, used in the second
step, makes \kltwod{} complete for all sets of rules. This does not follow
as stated. First,
the hypothesis of Prop.~16, that $\kl(R)\subseteq\deriv(R)$ for \emph{all}
sets $R$, is false under the finite-input reading of~\cite[\S3.2]{ESS20},
when $\Atoms$ is infinite, for every calculus with finite derivations that
is sound for finite rule sets (Remark~\ref{rem:infinite}); under the
arbitrary-input reading it is open. So the proposition cannot be applied as
stated. Second, even when $R$ is finite, the \klone{} rule set
$R\cup\Cons\cup\aux_e(R)$ in the middle step is infinite whenever $\Atoms$ is,
because $\Cons$ contains a constraint for every atom; finite completeness does
not apply to it directly. Third, Propositions~6 and~13 concern \klone{}
entailment for the infinite rule sets $R\cup\Cons\cup\aux(R)$ and
$R\cup\Cons\cup\aux_e(R)$, and their proofs in~\cite{ESS20} are brief: the
second rests on Proposition~12, whose proof is only sketched, and parts of
the development appeal to the least-model characterization, which fails for
some infinite rule sets (Remark~\ref{rem:infinite}).

For finite $R$ all three gaps can be closed. The following proposition
re-establishes the decomposition for finite rule sets; its proof, in
Appendix~\ref{app:kl2}, uses only the operational semantics.

\begin{proposition}[Decomposition for finite rule sets]\label{prop:decomposition}
For every finite set $R$ of \kltwo{} rules,
\[
\kltwofn(R)=\kl(R\cup\Cons\cup\aux(R))=\kl(R\cup\Cons\cup\aux_e(R)).
\]
\end{proposition}

\begin{lemma}\label{lem:auxe}
If $R$ is finite, then $v^*(R)$ and $\aux_e(R)$ are finite, and their
members contain only literals over atoms of $R$.
\end{lemma}

\begin{proof}
The members of $v(R)$ are subsets of the finite set of literals over atoms
of $R$, and resolution stays within it. Hence only finitely many sets, and
finitely many rules, arise.
\end{proof}

\begin{lemma}[Irrelevant constraints]\label{lem:irrelevant}
Let $R'$ be a set of \klone{} rules, $r$ a rule, and $L$ a finite set of
atoms containing every atom of $R'$ and $r$. Let $E$ be a set of
constraints $B\must\bot$ whose bodies are non-empty and disjoint from $L$.
If $R'\cup E\entails r$, then $R'\entails r$.
\end{lemma}

\begin{proof}
Let $R''$ be $R'$ or $R'\cup\{r\}$. Constraints never enable steps, so
$\out(R''\cup E,A)=\{X\in\out(R'',A)\mid X\models E\}$. For $A\subseteq L$
every $X\in\out(R'',A)$ is included in $L$ and thus satisfies $E$, so
$\out(R'',A)=\out(R''\cup E,A)$; with the hypothesis,
$\out(R',A)=\out(R'\cup\{r\},A)$ for $A\subseteq L$. For arbitrary $A$,
traces from $A$ and from $A\cap L$ correspond step by step, because all
bodies and heads of $R''$ lie within $L$, and satisfaction of $R''$ depends
only on atoms in $L$. Hence
$\out(R'',A)=\{X\cup(A\setminus L)\mid X\in\out(R'',A\cap L)\}$, and the
equality for $A\cap L$ yields the equality for $A$.
\end{proof}

\begin{theorem}[\kltwo, finite rule sets]\label{thm:kl2}
For every finite set $R$ of \kltwo{} rules, $\kltwofn(R)=\derivtwod(R)$. That
is, \kltwod{} is sound and complete for finite sets of rules.
\end{theorem}

\begin{proof}
\emph{Completeness.} Let $R\entailstwo r$. Let $S$ be the finite set of
atoms occurring in $R$ and $r$, let $L$ be the set of literals over $S$, and
split $\Cons$ into $\Cons_S=\{a\wedge\nt a\must\bot\mid a\in S\}$ and
$E=\Cons\setminus\Cons_S$. The set $R'=R\cup\Cons_S\cup\aux_e(R)$ is finite and its
rules, like $r$, only involve literals in $L$ (Lemma~\ref{lem:auxe}). By
Proposition~\ref{prop:decomposition}, $R'\cup E\entails r$. The constraints
in $E$ have the non-empty bodies $\{a,\nt a\}$ with $a\notin S$, disjoint
from $L$, so Lemma~\ref{lem:irrelevant} gives $R'\entails r$, and
Theorem~\ref{thm:completeness}, applied with $\Lits$ as the set of atoms,
gives $R'\derivesd r$. Every member of $\Cons_S$ is an instance of
$\nt$\textsc{-left}. Every member $A\setminus\{\ell\}\must\cmp{\ell}$ of
$\aux_e(R)$ is derivable in \kltwo{} from $R$~\cite[Prop.~14 and the proof
of Prop.~16]{ESS20}: $A\must\bot$ is obtained from rules of $R$ by the
derived rules \textsc{must}-$\bot$ and \textsc{resolve}-$\bot$
of~\cite[\S4]{ESS20}, and $\nt$\textsc{-right} yields the member. Both
derived rules can be obtained using only $\nt$\textsc{-left},
$\nt$\textsc{-right} and rules of Fig.~\ref{fig:rules} other than
\rn{may-trans}. Since \kltwod{}
contains \kld, it follows that $r\in\derivtwod(R)$.

\emph{Soundness.} Let $R_a=R\cup\Cons\cup\aux(R)$. By
Proposition~\ref{prop:decomposition} it suffices to show that $\kl(R_a)$
contains $R$ and is closed under the rules of \kltwod. It contains $R$, and
it is closed under the rules of \kld{} by Theorem~\ref{thm:soundness}, which
applies to the infinite set $R_a$. It contains the conclusions of
$\nt$\textsc{-left}, which belong to $\Cons$. For $\nt$\textsc{-right}, let
$A\wedge b\must\bot\in\kl(R_a)$. By Proposition~\ref{prop:local}, no model
of $R_a$ includes $A\cup\{b\}$, so $A\cup\{b\}$ is violating for $R$ by
Lemma~\ref{lem:reach2} in Appendix~\ref{app:kl2}. Hence
$A\setminus\{b\}\must\cmp{b}$ belongs to $\aux(R)$, and $A\must\cmp{b}$ is in
$\kl(R_a)$ by \rn{must-si} and Theorem~\ref{thm:soundness}.
\end{proof}

\begin{remark}[Scope]\label{rem:kl2-status}
(i)~Theorem~\ref{thm:kl2} and Proposition~\ref{prop:decomposition} are
proved on paper only; they are not mechanized, and their proofs still
require independent mathematical review. (ii)~The restriction to
finite $R$ cannot be dropped under the finite-input reading: for
$R_\infty$ as in Remark~\ref{rem:infinite}, $\outtwo(R_\infty,A)=\emptyset$
for all finite $A$, while every finite $R'\subseteq R_\infty$ has
$\outtwo(R',\emptyset)\neq\emptyset$ by Lemma~\ref{lem:violating}, so the
argument of Remark~\ref{rem:infinite}, with the soundness half of
Theorem~\ref{thm:kl2}, shows that $\top\must\bot$ is entailed but not
derivable. (iii)~Since \rn{may-trans} is derivable in \kltwod{} (the proof
of Proposition~\ref{prop:generalization}(a) applies verbatim), the original
\kltwo{} calculus is sound for finite rule sets. We have not determined
whether it is complete: the transfer of the counterexample by
conservativity~\cite[Prop.~18]{ESS20} is not immediate, since the syntactic
half of that proposition is derived from Prop.~17, which assumes the
completeness of \klone. (iv)~Propositions~6 and~13 of~\cite{ESS20} for
infinite $R$ are not addressed here.
\end{remark}

\subsection{\klthree}

\klthree{} extends \kltwo{} with variables, implicitly universally quantified
rules, and existential quantification in heads, which is interpreted by the
generation of fresh witnesses~\cite[\S5]{ESS20}. Our completeness proof does
not transfer automatically, because it depends on finiteness at three
points: the measure $\mu$, which bounds the number of atoms that can be
added to a body; the finiteness of successor sets, which makes
\rn{may-trans-disj} applicable; and Remark~\ref{rem:inputs}, which confines
attention to the atoms of $R$ and $r$. With existential witness generation
the set of relevant ground atoms is unbounded, the construction of
Lemma~\ref{lem:reach} need not terminate, and successor sets of ground
instances may be infinite. Moreover the inference rules of \klthree{} are not
those of \klone{} applied to ground instances, and \klthree{} is closely
related to geometric logic~\cite[\S5.5]{ESS20}. We therefore regard
completeness of \klthree{} as open. What the present results suggest is that
any complete calculus for \klthree{} will need a disjunctive transitivity
principle for permissives, and that a proof of completeness will need a
replacement for the finiteness arguments above.

% ---------------------------------------------------------------------------
\section{Philosophical Significance}
\label{sec:philosophy}
% ---------------------------------------------------------------------------

ESS locate the primary task of Kant's logic in an operational question:
given the rules an agent has adopted and the acts it is already performing,
which further acts may or must it perform? A secondary task is inferential:
given the rules adopted, which further rules may or must be
adopted~\cite[\S\S2.4, 3.6]{ESS20}. In \klone{} the first question is answered
by $\out$, the second by $\entails$ and by the calculus. Because $\entails$ is
defined through $\out$, the secondary question is answered by reference to
the primary one: a rule follows from $R$ just in case adopting it makes no
difference to what can be done from any starting point. A sound and complete
calculus shows that this consequence relation can be tracked by finitely
many forms of reasoning about rules alone, without running the procedures
that the rules describe.

The counterexample shows that the original forms of reasoning missed a
pattern implicit in those procedures. Under $\Rz$, the act $c$ can only be
performed together with $b$ or with $d$, each of which licenses it; neither
is required by $c$ on its own. What \rn{may-trans} cannot register is
\emph{grounding by alternatives}: every performance of $c$ is grounded in a
rule, though not always in the same one. This bears directly on the Kantian
requirement, emphasized by ESS in rejecting weakening of the output, that
activity be grounded in rules rather than arbitrary~\cite[\S3.7]{ESS20}.
Adopting $\top\may c$ licenses no arbitrary act, for every outcome
containing $c$ already contained a ground for it. \rn{may-trans-disj} makes
this explicit while preserving the requirement: its side condition still
demands that performing $c$ bring with it some act from which $c$ is
permitted, which is why the inference rejected in~\cite[\S3.5.1]{ESS20}
remains blocked. Theorem~\ref{thm:completeness} shows that, for finite sets
of rules, nothing else was missing.

Completeness matters for the internal coherence of the reconstruction. If
rules are procedures that govern mental activity, rather than declarative
propositions with truth-values, the inferential relations among rules must
be fixed by what the procedures do. Theorems~\ref{thm:soundness} and~\ref{thm:completeness}
show that, within their scope, reasoning about rules is a faithful proxy for
reasoning about the acts they govern: the two aspects of logic that ESS
distinguish do not come apart. Footnote~24 of~\cite{ESS20} suggests a
further significance, connecting completeness, via compactness, with the
passage from a proof-theoretic notion (consistency) to a world-directed one
(having a model), and relating this to the semantics of Achourioti and van
Lambalgen~\cite{AvL11}. Our results bear on the formal side of that
suggestion in two ways. Corollary~\ref{cor:model} is the expected
model-existence theorem for finite rule sets: a finite $R$ is consistent if
and only if it has an outcome. But Remark~\ref{rem:infinite} shows that the
compactness step is not available for $\out$: a countable set of rules may
have no outcome although each of its finite subsets has one, because
outcomes are finite packages of acts. An argument along the lines of
footnote~24 would thus have to work with infinite packages of acts, as in
the least-model semantics of~\cite[\S3.3]{ESS20}, which can diverge from
$\out$ for infinite rule sets. These are observations about the formal
reconstruction. They do not, by themselves, bear on the truth of
transcendental idealism or on the existence of a mind-independent world.

% ---------------------------------------------------------------------------
\section{Conclusion}
\label{sec:conclusion}
% ---------------------------------------------------------------------------

The calculus proposed for \klone{} in~\cite{ESS20} is incomplete. The
incompleteness has a single source, the singleton side condition of
\rn{may-trans}, and disappears when that condition is made disjunctive: the
resulting calculus \kld{} is sound for arbitrary sets of rules, complete for
finite ones, and contains the original calculus. Completeness cannot be
extended to infinite sets of rules under the operational semantics with
finite inputs. These results are machine-checked. For \kltwo, the general
form of Prop.~16 of~\cite{ESS20} cannot be used, but for finite rule sets
the decomposition into \klone{} can be re-established, and \kltwo{} with the
disjunctive rule is sound and complete; these proofs are on paper only and
still require independent review.
Open problems include the completeness of the original \kltwo{} calculus,
infinite rule sets under the arbitrary-input reading of Definition~5, and
the treatment of \klthree.

\paragraph{Acknowledgements.} We are grateful to Richard Evans (Google
DeepMind), whose work with Marek Sergot and Andrew Stephenson~\cite{ESS20}
introduced \klone, for correspondence and mathematical discussions, for
examining the counterexample and the proposed disjunctive repair, and for
sharing a preliminary AI-generated completeness argument for the repaired
calculus.

\paragraph{Use of AI tools.} The counterexample of
Section~\ref{sec:counterexample} and the rule \rn{may-trans-disj} are due to
Damiano Presa; the counterexample was found without AI assistance. The
preliminary completeness argument mentioned above was generated using
Google's Gemini. Anthropic's Claude Opus~5.5 (maximum effort) assisted in
the development and checking of the proofs, including those of
Section~\ref{sec:kl2} and Appendix~\ref{app:kl2}, and in the Agda
formalization. The mathematical claims and the formalized proofs remain the
responsibility of the authors.

\noindent\textbf{Data availability.} The Agda formalization will be
released as open-source code on GitHub. The files
\agda{KL1Completeness.agda} and \agda{KL1Supplement.agda} are available
from the authors on request.

\bibliographystyle{abbrvnat}
\bibliography{references}

% ---------------------------------------------------------------------------
\appendix
\section{Proofs for Section~\ref{sec:kl2}}
\label{app:kl2}
% ---------------------------------------------------------------------------

Throughout this appendix $R$ is a finite set of \kltwo{} rules,
$R_a=R\cup\Cons\cup\aux(R)$ and $R_e=R\cup\Cons\cup\aux_e(R)$. For a set $Y$ of
literals, $\mathit{at}(Y)$ is the set of atoms of its literals, and for a set
$P$ of atoms, $\mathit{Lit}(P)$ is the set of literals over $P$. As before, a
model of $R_a$ is a finite set of literals that satisfies $R_a$. A maximal
consistent set $X_m$ fails to satisfy $B\must C$ if and only if it includes
$B\cup\cmp{C}$, since $C\cap X_m=\emptyset$ amounts to
$\cmp{C}\subseteq X_m$; and it cannot include an inconsistent set. Hence
$X_m\models R$ if and only if $X_m$ includes no member of $v(R)$.

\begin{lemma}[Flipping]\label{lem:flip}
\begin{enumerate}[label=(\alph*)]
\item Let $W$ be a finite set of literals that is not violating for $R$,
  and $\ell$ a literal with $\mathit{at}(\ell)\notin\mathit{at}(R)\cup
  \mathit{at}(W)$. Then $W\cup\{\ell\}$ is not violating for $R$.
\item Every maximal consistent set that satisfies $R$ satisfies $\Cons$ and
  $\aux(R)$.
\end{enumerate}
\end{lemma}

\begin{proof}
(a) Let $X_m\supseteq W$ be maximal consistent with $X_m\models R$, and let
$X_m'$ be obtained from $X_m$ by replacing its literal over
$\mathit{at}(\ell)$ by $\ell$. Then $X_m'$ is maximal consistent, includes
$W\cup\{\ell\}$ because $W$ contains no literal over $\mathit{at}(\ell)$, and
satisfies $R$ because $X_m$ and $X_m'$ agree on all literals over atoms of
$R$. (b) A maximal consistent $X_m$ is consistent, so it satisfies $\Cons$. Let
$V\setminus\{\ell\}\must\cmp{\ell}$ be a rule of $\aux(R)$ with
$V\setminus\{\ell\}\subseteq X_m$. If $\ell\in X_m$, then $V\subseteq X_m$,
which is impossible because $V$ is violating; so $\cmp{\ell}\in X_m$ by
maximality.
\end{proof}

\begin{lemma}[Reaching outcomes]\label{lem:reach2}
If a finite set $X$ of literals is not violating for $R$, then
$\outtwo(R,X)=\out(R_a,X)$ contains a set $T$ with
$X\subseteq T\subseteq\mathit{Lit}(\mathit{at}(R)\cup\mathit{at}(X))$.
Consequently, if no model of $R_a$ includes $X$, then $X$ is violating
for $R$.
\end{lemma}

\begin{proof}
Fix a maximal consistent $X_m\supseteq X$ with $X_m\models R$; by
Lemma~\ref{lem:flip}(b), $X_m$ satisfies every rule of $R_a$. Let
$L_0=\mathit{Lit}(\mathit{at}(R)\cup\mathit{at}(X))$. Starting from $T=X$, as
long as some imperative of $R_a$ is violated by $T$, add to $T$ a literal of
its head that belongs to $X_m$. We maintain $T\subseteq X_m\cap L_0$ and
$T\in\cns(R_a,X)$. Since $X_m$ satisfies the violated imperative and
includes its body, its head contains some $h\in X_m$, and $h\notin T$; in
particular the imperative is not a constraint. If it belongs to $R$, then
$h\in L_0$. Otherwise it is a rule $V\setminus\{\ell\}\must\cmp{\ell}$ of
$\aux(R)$ and $h=\cmp{\ell}$; as $V\setminus\{\ell\}\subseteq T\subseteq
L_0$, if $\mathit{at}(\ell)$ were not in $\mathit{at}(R)\cup\mathit{at}(X)$,
Lemma~\ref{lem:flip}(a) applied to $W=V\setminus\{\ell\}\subseteq X_m$ would
show that $V$ is not violating. So $h\in L_0$ in all cases. As
$X_m\cap L_0$ is finite, the construction stops, with a set $T$ that
satisfies every imperative of $R_a$, that is, $T\in\out(R_a,X)$. The last
claim follows because $T$ is a model of $R_a$ that includes $X$.
\end{proof}

\begin{lemma}\label{lem:violating}
For every finite set $X$ of literals, $\outtwo(R,X)=\emptyset$ if and only if
$X$ is violating for $R$.
\end{lemma}

\begin{proof}
If $X$ is not violating, Lemma~\ref{lem:reach2} gives an element of
$\outtwo(R,X)$. Conversely, let $X$ be violating and suppose
$T\in\outtwo(R,X)$. Then $T$ is finite, consistent (it satisfies $\Cons$) and
violating (it includes $X$). If $T\neq\emptyset$, pick $\ell\in T$: the rule
$T\setminus\{\ell\}\must\cmp{\ell}$ of $\aux(R)$ forces $\cmp{\ell}\in T$,
contradicting consistency. If $T=\emptyset$, then for any atom $a$ the sets
$\{a\}$ and $\{\nt a\}$ are violating, so $\top\must\nt a$ and $\top\must a$
belong to $\aux(R)$ and cannot both be satisfied by a consistent $T$.
\end{proof}

\begin{lemma}[Resolution; cf.~{\cite[Prop.~12]{ESS20}}]\label{lem:resolution}
Every member of $v^*(R)$ is violating for $R$, so that
$\aux_e(R)\subseteq\aux(R)$. Conversely, every consistent finite set of
literals that is violating for $R$ includes a member of $v^*(R)$.
\end{lemma}

\begin{proof}
A maximal consistent superset of $B\cup\cmp{C}$ does not satisfy $B\must C$,
so the members of $v(R)$ are violating. If $A\cup\{\ell\}$ and
$B\cup\{\cmp{\ell}\}$ are violating, so is $A\cup B$: a maximal consistent
$X_m\supseteq A\cup B$ contains $\ell$ or $\cmp{\ell}$, and therefore
includes one of the two sets. This proves the first claim.

For the second, let $\Gamma=v^*(R)$, a finite set by Lemma~\ref{lem:auxe},
and let $V$ be consistent and violating. Every maximal consistent superset
of $V$ includes a member of $v(R)\subseteq\Gamma$. We argue by induction on
the number $k$ of atoms that occur in $\Gamma$ and on which $V$ is
undecided, i.e.\ such that $V$ contains neither the atom nor its negation.
If $k=0$, let $X_m$ be a maximal consistent superset of $V$; it includes some
$W\in\Gamma$, and since every literal of $W$ is over an atom decided by $V$,
on which $X_m$ agrees with $V$, we have $W\subseteq V$. If $k>0$, pick such
an atom $p$. The sets $V\cup\{p\}$ and $V\cup\{\nt p\}$ are consistent,
violating, and undecided on $k-1$ atoms of $\Gamma$, so by induction they
include members $W_1$ and $W_2$ of $\Gamma$. If $p\notin W_1$ or
$\nt p\notin W_2$, then $W_1\subseteq V$ or $W_2\subseteq V$. Otherwise
$W_1=A\cup\{p\}$ and $W_2=B\cup\{\nt p\}$ with $A,B\subseteq V$, and their
resolvent $A\cup B\in\Gamma$ is included in $V$.
\end{proof}

\begin{proof}[Proof of Proposition~\ref{prop:decomposition}]
(1)~$\kltwofn(R)\subseteq\kl(R_a)$. Let $R\entailstwo r$. If $V$ is a finite
violating set for $R\cup\{r\}$, then $\outtwo(R\cup\{r\},V)=\emptyset$ by
Lemma~\ref{lem:violating} applied to the finite set $R\cup\{r\}$, hence
$\outtwo(R,V)=\emptyset$, and $V$ is violating for $R$ by
Lemma~\ref{lem:violating}. The converse is immediate, so
$\aux(R\cup\{r\})=\aux(R)$, and for every $A$,
$\out(R_a\cup\{r\},A)=\outtwo(R\cup\{r\},A)=\outtwo(R,A)=\out(R_a,A)$.

(2)~$\kl(R_a)\subseteq\kltwofn(R)$. Let $R_a\entails r$ with $r=B\anyto H$.
As in (1), it suffices to show that every finite violating set $V$ for
$R\cup\{r\}$ is violating for $R$. This is immediate if $r$ is a permissive.
Let $r$ be an imperative, and suppose that $X_m\supseteq V$ is maximal
consistent with $X_m\models R$; then $X_m\not\models r$, so $B\subseteq X_m$
and $H\cap X_m=\emptyset$. Let $T=X_m\cap\mathit{Lit}(\mathit{at}(R)\cup
\mathit{at}(r))$, a finite set that includes $B$ and satisfies $R$ and $\Cons$.
It also satisfies $\aux(R)$: if $W\setminus\{\ell\}\must\cmp{\ell}$ is a rule
of $\aux(R)$ with $W\setminus\{\ell\}\subseteq T$, then
$\cmp{\ell}\in X_m$ by Lemma~\ref{lem:flip}(b), and
$\mathit{at}(\ell)\in\mathit{at}(R)\cup\mathit{at}(W\setminus\{\ell\})
\subseteq\mathit{at}(R)\cup\mathit{at}(r)$ by Lemma~\ref{lem:flip}(a), so
$\cmp{\ell}\in T$. Thus $T$ is a model of $R_a$ that includes $B$ and does not
meet $H$, contradicting $\Valid_{R_a}(B,H)$, which holds by
Proposition~\ref{prop:local}.

(3)~$\out(R_a,A)=\out(R_e,A)$ for all $A$; then $\kl(R_a)=\kl(R_e)$ by
Corollary~\ref{cor:invariance}. By Lemma~\ref{lem:resolution},
$R_e\subseteq R_a$. We first show that $R_e\derives W\must\bot$ for every
$W\in v^*(R)$, in the original calculus of Fig.~\ref{fig:rules}. If
$W\neq\emptyset$, pick $w\in W$: from $W\setminus\{w\}\must\cmp{w}$ in
$\aux_e(R)$, \rn{must-si} gives $W\must\cmp{w}$; from the member of $\Cons$ for
the atom of $w$, \rn{must-si} gives $W\wedge\cmp{w}\must\bot$; and
\rn{must-trans} gives $W\must\bot$. If $W=\emptyset$, either
$\top\must\bot\in R$, or $\emptyset$ is the resolvent of some $\{\ell\}$ and
$\{\cmp{\ell}\}$ in $v^*(R)$; then $\top\must\cmp{\ell}$ and $\top\must\ell$
belong to $\aux_e(R)$, \rn{must-si}, the member of $\Cons$ for the atom of $\ell$
and \rn{must-trans} give $\ell\must\bot$, and \rn{must-trans} applied to
$\top\must\ell$ gives $\top\must\bot$.

Next, every rule $\rho=V\setminus\{\ell\}\must\cmp{\ell}$ of $\aux(R)$ is
derivable from $R_e$. If $V$ contains a pair $a,\nt a$ and both lie in
$V\setminus\{\ell\}$, then $\rho$ follows from $a\wedge\nt a\must\bot$ by
\rn{must-si} and \rn{quod-libet}; if one of them is $\ell$, then
$\cmp{\ell}\in V\setminus\{\ell\}$, and $\rho$ follows by \rn{must-id} and
\rn{must-si}. If $V$ is consistent, Lemma~\ref{lem:resolution} gives some
$W\in v^*(R)$ with $W\subseteq V$. If $\ell\in W$, then $\rho$ follows from
$W\setminus\{\ell\}\must\cmp{\ell}\in\aux_e(R)$ by \rn{must-si}; otherwise
$W\subseteq V\setminus\{\ell\}$, and $\rho$ follows from $W\must\bot$ by
\rn{must-si} and \rn{quod-libet}.

By Theorem~\ref{thm:soundness}, which applies to the infinite set $R_e$,
every $\rho\in\aux(R)$ is entailed by $R_e$. Let $X\in\out(R_a,A)$. Its
trace uses only finitely many rules $\rho_1,\dots,\rho_k$ of
$\aux(R)\setminus\aux_e(R)$, so
$X\in\out(R_e\cup\{\rho_1,\dots,\rho_k\},A)=\out(R_e,A)$, by $k$ applications
of Corollary~\ref{cor:invariance}. Conversely, if $X\in\out(R_e,A)$, then
$X\in\cns(R_a,A)$, and $X$ satisfies each $\rho\in\aux(R)$ because
$X\in\out(R_e\cup\{\rho\},A)$. Hence $X\in\out(R_a,A)$.
\end{proof}

\end{document}